\documentclass[11pt,letterpaper]{article}
\usepackage[margin=1in]{geometry}

\usepackage{amsmath}
\usepackage[noend]{algpseudocode}
\usepackage{algorithm}
\usepackage{amsthm}
\usepackage{amsfonts}
\usepackage{graphicx}
\usepackage{hyperref}
\usepackage{xcolor}
\usepackage{xspace}
\usepackage{bm}
\usepackage{enumitem}
\usepackage{authblk}

\newtheorem{theorem} {Theorem}

\newtheorem{proposition}[theorem] {Proposition}

\newcommand{\calD}{\mathcal{D}}
\newcommand{\harddistrib}{\calD_{\mathrm{hard}}}

\newcommand{\poly}{\operatorname{poly}}
\let\eps\varepsilon
\let\epsilon\varepsilon

\newcommand{\tdigest}{$t$-digest\xspace}
\newcommand{\reqsketch}{\textsf{ReqSketch}\xspace}

\begin{document}

\title{Theory meets Practice at the Median: \\a worst case comparison of relative error quantile algorithms}

\author[1]{Graham Cormode\thanks{Supported by European Research Council grant ERC-2014-CoG 647557.}}
\author[2]{Abhinav Mishra}
\author[2]{Joseph Ross}
\author[3]{Pavel Vesel{\'{y}}\thanks{Work done while the author was at the University of Warwick.
Partially supported by European Research Council grant ERC-2014-CoG 647557, by GA \v{C}R project 19-27871X and by Charles University project UNCE/SCI/004.}}

\affil[1]{University of Warwick, \texttt{G.Cormode@warwick.ac.uk}}
\affil[2]{Splunk, \texttt{\{amishra,josephr\}@splunk.com}}
\affil[3]{Charles University, \texttt{vesely@iuuk.mff.cuni.cz}}

\date{}

\maketitle

\begin{abstract}
Estimating the distribution and quantiles of data is a foundational task in data mining and data science. 
We study algorithms which provide accurate results for extreme quantile queries using a small amount of space, thus helping to understand the tails of the input distribution.
Namely, we focus on two recent state-of-the-art solutions: \tdigest and \reqsketch.
While $t$-digest is a popular compact summary which works well in a variety of settings,
\reqsketch comes with formal accuracy guarantees at the cost of its size growing as new observations are inserted.
In this work, we provide insight into which conditions make one preferable to the other.
Namely, we show how to construct inputs for \tdigest that induce an almost arbitrarily large error
and demonstrate that it fails to provide accurate results even on i.i.d.\ samples from a highly non-uniform distribution.
We propose practical improvements to \reqsketch, making it faster than \tdigest,
while its error stays bounded on any instance.
Still, our results confirm that \tdigest remains more accurate on the ``non-adversarial'' data encountered in practice.
\end{abstract}

\section{Introduction}

Studying the distribution of data is a foundational task in data mining and data science. 
Given observations from a large domain, we will often want to track the cumulative frequency distribution, to understand the behavior, or to identify anomalies. %
This cumulative distribution function (CDF) is also known variously as the order statistics, generalizing the median, and the quantiles.  
When we have a very large number of input observations, an exact characterization is excessively large, and we can be satisfied with an approximate representation, i.e., a compact function whose distance from the true CDF is bounded.  
Recent work has argued that, rather than a uniform error bound, it is more important to capture the detail of the tail of the input distribution. 

Faced with the problem of processing large volumes of distribution data, there have been many proposals of approximate quantile algorithms to extract the desired compact summary.  
These are designed to handle the input when seen as a stream of updates, or as distributed observations.
Even though these various algorithms all draw on the same set of motivations, the emphasis can vary widely. 
Some works view the question primarily as one of computational complexity, and seek optimal bounds on the space usage, even if this entails very intricate algorithmic designs and lengthy technical proofs. 
Other works aspire to highly practical algorithms that can be implemented and run efficiently on real workloads. 
Although their authors might object, we can crudely characterize these two perspectives as ``theoretically-driven'' and ``pragmatic''.  

In this paper, we study the behavior of two recent algorithms for the quantiles problem, which we take to embody these two mindsets.  
The pragmatic approach is represented by the \tdigest, which is a flexible framework that has been reported as being adopted in practice by various tech-focused companies (e.g., Microsoft, Facebook, Google~\cite{dunning21}). 
The theoretical approach is represented by the \reqsketch~\cite{cormode2020relative}, a work building on a line of prior theoretical papers, each making incremental improvements to the asymptotic bounds.  

On first glance, the conclusion seems obvious. 
\reqsketch is suited for algorithmic study, and contributes to our understanding of the fundamental computational complexity of the problem.  
Coding it up is not too hard, but the constants hidden in the ``big-O'' analysis mean that it requires a fair amount of space to store and so is unlikely to be competitive with the pragmatic approach. 
Meanwhile, the \tdigest is very compact, and gives accurate answers to realistic workloads, especially those uniformly distributed over the domain.
Its widespread adoption should give confidence that this is a sensible choice to implement. 

Our contribution in this paper is to tell a more nuanced story, with a less clearcut ending. 
We dive into the inner workings of the \tdigest, and show how to construct inputs that lead to almost arbitrarily bad accuracy levels. 
While it may seem that such inputs are highly unlikely to be encountered in practice,
\tdigest may fail to provide accurate estimates even if input items are repeatedly drawn from a non-uniform distribution, and we demonstrate such distributions.
Meanwhile, we engineer an implementation of \reqsketch that improves its time and space efficiency,
making it faster than \tdigest.
The outcome is a collection of empirical results showing \reqsketch can be vastly preferable to \tdigest, even on i.i.d. samples, flipping the conclusion for uniformly distributed inputs. %

Still, the conclusion is therefore less straightforward than one might wish for.  
Both the input distributions and the careful construction which lead to high error for the \tdigest rely on a highly non-uniform data distribution 
with numbers ranging from infinitesimal to astronomically large.
For most realistic data patterns encountered in practice, the \tdigest will remain a compelling choice, due to its simplicity and ease of use. 
But, particularly since quantile monitoring is often needed to track deviations from expected behavior, there is now a case to adopt \reqsketch for scenarios where a strong guarantee is needed across all eventualities, no matter how unlikely they might appear.  
Although \reqsketch has higher overheads on average, and relies on internal randomness, its worst case is much more tightly bounded.  
So, in summary, the practical approach nevertheless has flaws -- at least in theory -- while the theoretical approach is not so impractical as it may first appear.

\paragraph{Paper outline.}
We start with necessary definitions and a brief review of related work in the next section.
Then we describe both \tdigest and \reqsketch in Section~\ref{sec:algs}, where we also provide a short summary of practical improvements of the latter sketch.
In Section~\ref{sec:attack}, we outline a worst-case input construction, showing that \tdigest may suffer an almost arbitrarily large error.
Empirical results appear in Section~\ref{sec:empirical}, where we demonstrate the error behavior of both algorithms on the aforementioned worst-case inputs
and also on inputs that consist of i.i.d.\ samples from a distribution. Finally, we provide a comparison of average update times in Section~\ref{sec:updateTime}.

\section{Preliminaries}

\subsection{Definitions}
\label{sec:defs}

We consider algorithms that operate on a stream of items drawn from some large domain. 
This could be any domain $U$ equipped with a total order (e.g., strings with lexicographic comparison), or a more restricted setting, such as the reals, where we additionally have arithmetic operations. 

The core notion needed is that of the \textit{rank} of an element from the domain, which is the number of items from the input that are (strictly) smaller than the given element.  
Formally, for an input stream of $n$ items $\sigma = \{x_1, \ldots , x_n\}$, the rank of element $y$ is given by 
$R_\sigma(y) = |\{ i : x_i < y\}|$. 
There is some nuance in how to handle streams with duplicated elements, but we will gloss over it in this presentation; see~\cite{DataSketches-Quantiles-Definitions} for a discussion of this nuance. 

The quantiles are those elements which achieve specific ranks. 
For example, the median of $\sigma$ is $y$ such that $R_\sigma(y) = n/2$, and the $p$th percentile is $y$ such that $R_\sigma(y) = pn/100$. 
More generally, we seek the $q$th quantile with $R_\sigma(y) = qn$ for $0\leq q \le 1$.
Again, ambiguity can arise since there can be a range of elements satisfying this definition, but this need not concern us here.  

We consider algorithms that aim to find approximate quantiles via approximate ranks. 
That is, they seek elements whose rank is sufficiently close to the requested quantile. 
Specifically, the \textit{quantile error} of reporting $y$ as the $q$th quantile is given by $|q - R_{\sigma}(y)/n|$.
A standard observation is that if we have an algorithm to find the approximate rank of an item, as $\hat{R}_{\sigma}(y)$, this suffices to answer quantile queries, after accounting for the probability of making errors. 
In what follows, we focus on the accuracy of rank estimation. 

Typically, we would like to have some guarantee on rank estimation accuracy.  
A uniform (additive) rank estimation guarantee asks that the error on all queries be bounded by the same fraction of the input size, i.e., $|\hat{R}_{\sigma}(y) - R_{\sigma}(y)| \leq \epsilon n$, for $\epsilon < 1$. 
This will ensure that all quantile queries have the same accuracy. 
However, it is noted that in practice we often want greater accuracy on the tails of the distribution, where we can see more variation, compared to the centre which is usually more densely packed and unvarying. 
This leads to notions such as the relative error guarantee, $|\hat{R}_{\sigma}(y) - R_{\sigma}(y)| \leq \epsilon R_{\sigma}(y)$, or more generally 
$|\hat{R}_{\sigma}(y) - R_{\sigma}(y)| \leq \epsilon f(R_\sigma(y)/n)$, 
where $f$ is a \textit{scale function} which captures the desired error curve as a function of the location in quantile space.  
In this work, we focus on the relative error guarantee and related scale functions (based on logarithmic functions).
As defined, relative error focuses on the low end of the distribution (i.e., the elements with low rank), but it is straightforward to flip this to the high end by using the scale function $(1 - R_{\sigma}(y)/n)$, or to make it symmetric with the scale function $\min(R_{\sigma}(y)/n, 1-R_{\sigma}(y)/n)$. 

We remark that guarantees on the rank error stated above are more general and mathematically natural to study than guarantees on the value error, which requires, for example, for a quantile query $q$ to return an item $y$ with $y = (1\pm \epsilon)\cdot x$, where $R(x) = q$. The error in value space is not invariant under translating the data (and can be made arbitrarily bad by applying a translation). For applications in which value space guarantees are desired, a simple logarithmic histogram is optimal in the fixed-range regime.

\subsection{Related Work}
\label{sec:related}

Most related work has focused on providing uniform error guarantees. 
It is folklore that a random sample of size $O(\frac{1}{\epsilon^2})$ items from the input is sufficient to provide an $\epsilon$ additive error estimate for any quantile query, with constant probability. 
Much subsequent work has aimed to improve these space bounds. 
Munro and Paterson~\cite{munro1980selection} gave initial results, but it was not until two decades later that Manku et al. reinterpreted this (multipass) algorithm as a quantile summary taking one pass over a stream, and showed improved bounds of $O(\frac{1}{\epsilon}\log^2 \epsilon n)$~\cite{manku1998approximate}. 
For many years, the state of the art was the Greenwald-Khanna (GK) algorithm, which comes in two flavours: a somewhat involved algorithm with a $O(\frac{1}{\epsilon} \log \epsilon n)$ space cost, and a simplified version without a formal guarantee, but good behavior in practice~\cite{greenwald2001space}. Subsequent improvements came from Felber and Ostrovsky~\cite{felber2015randomized}, who proposed combining sampling with a constant number of GK instances; and Agarwal et al. \cite{agarwal2013mergeable} who adapted the Manku et al. approach with randomness. 
Both these approaches removed the (logarithmic) dependence on $n$  from the space cost. 
Most recently, Karnin et al.~\cite{karnin2016optimal} (KLL) further refined the randomized approach to show an $O(\frac{1}{\epsilon})$ space bound. 
The tightest bound is achieved by a more complicated variant of the approach; a slightly simplified approach with weaker guarantees is implemented in the Apache DataSketches library~\cite{datasketches}. %

The study of other scale functions such as relative error can be traced back to the work of Gupta and Zane~\cite{Gupta:2003:CIL:644108.644150}, who gave a simple multi-scale sampling-based approach, with a space bound of $O(\frac{1}{\epsilon^3} \poly \log n)$. 
Subsequent heuristic and theoretical work aimed to reduce this cost, leading to a bound of $O(\frac{1}{\epsilon} \log^3 (\epsilon n))$ using a deterministic merge \& prune strategy due to Zhang and Wang~\cite{zhang2006space}. 
The cubic dependence on $\log n$ can be offputting, and recent work in the form of the \reqsketch~\cite{cormode2020relative} has reduced this bound by adopting a randomized algorithm inspired by the KLL method. 

The theoretical study of the quantiles problem has also led to lower bounds on the amount of space needed by any algorithm for the problem, based on information theoretic arguments about how many items from the input need to be stored. 
A simple argument shows that a uniform error guarantee requires space to store $\Omega(1/\epsilon)$ items from the input, and a relative error requires
$\Omega(1/\epsilon \log \epsilon n)$ (see, e.g., \cite{cormode2020relative}). 
Some more involved arguments show stronger lower bounds for uniform error of
$\Omega(\frac{1}{\epsilon} \log 1/\epsilon)$~\cite{hungting10}
and $\Omega(\frac{1}{\epsilon} \log \epsilon n)$~\cite{cormode2019tight}, but only for \textit{deterministic} and \textit{comparison-based} algorithms. 
The restriction to comparison-based means that the method can only apply comparisons to items and is not permitted to manipulate them (e.g., by computing the average of a set of items). 
Nevertheless, these lower bounds are sufficient to show that the analysis of certain approaches described above, such as the GK algorithm, is asymptotically tight, and cannot be improved further. 
The deterministic bounds can also be extended to apply to randomized algorithms and non-uniform guarantees, becoming weaker as a result. 

There are other significant approaches to the quantiles problem to study. 
The moment-based sketch takes a statistical approach, by maintaining the moments of the input stream (empirically), and using these to fit the maximum entropy distribution which agrees with these moments~\cite{gan2018moment}. 
This requires the assumption that the model fitting procedure will yield a distribution that closely agrees with the true distribution. 
The DDSketch aspires to achieve a ``relative error'' guarantee \textit{in value space}~\cite{masson2019ddsketch}, as described above, though its merge operation (needed to handle the unbounded range case) may result in estimates that do not comply with the prescribed accuracy.
Finally, the \tdigest~\cite{tdigest-origPaper,dunning21} has been widely used in practice, and is described in more detail in the subsequent section.

\section{Algorithms }
\label{sec:algs}

\subsection{t-Digest}
\label{sec:tdigest}

The \tdigest consists essentially of a set of weighted \emph{centroids} $\{ C_1, C_2, \ldots \}$, with a weighted centroid $C_i = (c_i, w_i)$ representing $w_i \in \mathbb{Z}$ points near $c_i \in \mathbb{R}$. Centroids are maintained in the sorted order, that is, $c_i < c_j$ for $i < j$. Rank queries are answered approximately by accumulating the weights smaller than a query point, and performing linear interpolation between the straddling centroids. The permissible weight of a centroid is governed by a non-decreasing \textit{scale function} $k : [0, 1] \to \mathbb{R} \cup \{ \pm \infty \}$, which describes the maximal centroid weight as a function on quantile space: faster growth of $k$ enforces smaller centroids and hence higher accuracy. In particular, scale functions which grow rapidly near the tails $q =0, 1$ but are flat near $q = 0.5$ should produce accurate quantile estimates near $q=0, 1$, but trade accuracy for space near $q=0.5$.

The size of a \tdigest is controlled by a \emph{compression parameter} $\delta$, which (roughly) bounds from above the number of centroids used. (For the scale functions below, Dunning \cite{dunning2019size}  shows that this rough bound does hold for all possible inputs.)
Given $\delta$ and scale function $k$, the weight $w_i$ of centroid $C_i$ must satisfy
\begin{equation}\label{eqn:weight_limit}
    k\left(\frac{w_{<i} + w_i}{N}\right) - k\left(\frac{w_{<i}}{N}\right) \le \frac{1}{\delta}\,,
\end{equation}
where $w_{<i} := \sum_{j < i} w_j$ is the total weight of centroids to the left of $C_i$ and $N = \sum_j w_j$ is the total number of items summarized by the \tdigest.
The intuitive meaning of~\eqref{eqn:weight_limit} is that the size of $C_i$ is determined by 
the inverse derivative of the scale function evaluated at the fraction of items to the left of $C_i$. Thus $\delta$ controls the size-accuracy tradeoff, and the scale function $k$ allows for the accuracy to vary across quantile space. We will say a \tdigest is associated with the pair $(k, \delta)$.
The four common proposed scale functions are
\begin{align*}
    k_0(q) &= \frac{q}{2}   \quad\quad\quad\quad\quad\quad
    k_1(q) = \frac{1}{2\pi}\sin^{-1}(2q-1) \\
    k_2(q) &= \frac{1}{Z(N)}\log\frac{q}{1-q} \\
    k_3(q) &= \frac{1}{Z(N)}\begin{cases}
                                \log 2q & \text{if } q\le 0.5 \\
                                -\log 2(1 - q) & \text{if } q > 0.5
                            \end{cases}
\end{align*}
Here, $Z(N)$ is a normalization factor that depends on $N$. 
While $k_0$ provides a uniform weight bound for any $q\in [0,1]$, functions $k_1, k_2,$ and $k_3$ get steeper towards the tails $q = 0, 1$, which
leads to smaller centroids and higher expected accuracy near $q = 0, 1$.
Dunning~\cite{dunning2019conservation} proves that adding more data to a \tdigest or merging two instances of \tdigest
preserves the constraint~\eqref{eqn:weight_limit} if any of these four scale functions is used.
Ross~\cite{ross2020asymmetric} describes asymmetric variants of $k_1, k_2,$ and $k_3$, using the given function $k$ on $[\alpha, 1]$ and the linearization of $k$ at $\alpha$ on $[0, \alpha)$, and shows that the \tdigest associated with any of these modified scale functions accepts insertions and is mergeable. %

There are two main implementations of \tdigest that differ in how they incorporate an incoming item into the data structure.
The \emph{merging} variant maintains a buffer for new updates and once the buffer gets full, it performs a merging pass,
in which it treats all items in the buffer as (trivial) centroids, sorts all centroids, and merges iteratively any two consecutive centroids
whose combined size does not violate the constraint~\eqref{eqn:weight_limit}.
The \emph{clustering} variant finds the closest centroids to each incoming item $x$ and adds $x$ to a randomly chosen one of the closest centroids to $x$ 
that still has room for $x$, i.e, satisfies~\eqref{eqn:weight_limit} after accepting $x$.
If there is no such centroid, the incoming item forms a new centroid, which may however lead to exceeding
the limit of $\delta$ on the number of centroids --- in such a case, we perform 
 a merging pass over the centroids that is guaranteed to output at most $\delta$ centroids.

In an ideal scenario, the instance of the \tdigest would be strongly-ordered, that is, for each $x_i$ represented by centroid $C_i$ and each $x_j$ represented by $C_j$ with $i < j$,
it holds that $x_i < x_j$. This means that the (open) intervals spanned by data points summarized by each centroid are disjoint, which is the case when data are presented in the sorted order.
Together with assuming a uniform distribution of items across the domain, strongly-ordered \tdigest provides highly accurate rank estimates even for, say, $\delta = 100$.
However, strong ordering of centroids is impossible to maintain in a limited memory when items arrive in an arbitrary order
and in general, centroids are just \emph{weakly ordered}, i.e.,
only the means $c_i$ and $c_j$ of centroids $C_i$ and $C_j$ satisfy $c_i < c_j$ if $i < j$.
This weak ordering of centroids, together with non-uniform distribution of items, is the major cause of the error in rank estimates.
The hard instances and distributions presented in this paper are constructed so that they induce a ``highly weak'' ordering of centroids, meaning
that many values summarized by centroid $C_i$ will not lie between the means of neighboring centroids. %
As we show below, this leads to a highly biased error in rank estimates for certain inputs.

\subsection{\reqsketch}
\label{sec:reqsketch}

The basic building block of \reqsketch~\cite{cormode2020relative} is a \emph{compactor}, which is essentially a buffer of a certain capacity for storing items,
and the sketch consists of several such compactors, arranged in levels numbered from $0$.
At the beginning, we start with one buffer at level $0$, which accepts incoming items.
Once the buffer at any level $h$ gets full, we discard some items from the sketch in a way that does not affect rank estimates too much.
This is done by sorting the buffer, choosing an appropriate prefix of an even size, and removing all items in the chosen prefix from the buffer.
Of these removed items, a randomly chosen half is inserted into the compactor at level $h+1$ (which possibly needs to be initialized first),
namely, items on odd or even indices with equal probability, while removed items in the other half are discarded from the sketch.
This procedure is called the \emph{compaction operation}.
Similar compactors were already used to design the KLL sketch~\cite{kane2010optimal} and appear also, %
e.g.,\ in~\cite{manku1998approximate,manku1999random,agarwal2013mergeable,luo16_quantiles_experimental}.

Since the sketch consists of items stored in the compactors, one can view the set of stored items as a weighted coreset, that is,
a sample of the input where each item stored at level $h$ is assigned a weight of $2^h$ (akin to the weight of a centroid in the \tdigest setting). 
Observe that the total weight of items remains equal to the input size: when a compaction operation discards $r$ items, it also promotes $r$ items one level higher,
and as the weight of promoted items doubles, the total weight of stored items remains unchanged after performing a compaction.
To estimate the rank of some query item $y$, we simply calculate the total weight of stored items $x$ with $x < y$.
Overall, this approach gives a comparison-based algorithm and thus, its behavior only depends on the ordering of the input and
is oblivious to applying an arbitrary order-preserving transformation of the input, even non-linear (which does not hold for \tdigest).

The error in the rank estimate for any item is unbiased, that is, 0 in expectation, and since it is a weighted sum
of a bounded number of independent uniform $\pm 1$ random variables, its distribution is sub-Gaussian, so we can apply standard tail bounds for the Gaussian distribution.
For bounding the variance, the choice of the prefix in the compaction operation is crucial. For uniform error, it is sufficient to always compact the whole buffer; see e.g.\ \cite{karnin2016optimal}.
However, as argued in~\cite{cormode2020relative}, to achieve relative error accuracy, the prefix size should be at most half of the buffer size and chosen according to an
exponential distribution, i.e., with probability proportional to $\exp(-s / k)$, where $s$ is the prefix size and $k$ is a parameter of the sketch controlling
the accuracy. The prefix is actually chosen according to a derandomization of this distribution, which leads to a cleaner analysis and smaller constant factors.
The choice of prefix is qualitatively similar to the choice of scale function for a \tdigest.

The number of compactors is bounded by at most $O(\log(N/B))$, where $B$ is the buffer size,
since the weight of items is exponentially increasing with each level, %
and thus, the size of the sketch is $O(B\cdot \log(N/B))$. The analysis in~\cite{cormode2020relative} implies that if we take $B = O(\varepsilon^{-1}\cdot \sqrt{\log \eps N})$,
then the sketch provides rank estimates with relative error $\eps$ with constant probability, while its size is $O(\varepsilon^{-1}\cdot \log^{1.5} \eps N)$
(the parameter $k$ mentioned above should be set to $O(B / \log \eps N) = O(\varepsilon^{-1}/ \sqrt{\log \eps N})$).
Finally, \reqsketch is fully mergeable, meaning that after an arbitrary sequence of merge operations, the aforementioned accuracy-space trade-off
still holds.

\subsection{Implementation Improvements of \reqsketch}
\label{sec:reqSketchImprovements}

The brief description of \reqsketch above follows the outline in~\cite{cormode2020relative} and is suitable for a mathematical analysis.
In this section, we describe practical adjustments to \reqsketch that improve constant factors as well as the running time.
These were used in our proof-of-concept Python implementation and have been incorporated in the implementation in the DataSketches library.\footnote{
The Python implementation of \reqsketch by the fourth author is available at \url{https://github.com/edoliberty/streaming-quantiles/blob/master/relativeErrorSketch.py}.
DataSketches library is available at \url{https://datasketches.apache.org/};
\reqsketch is implemented in this library according to the aforementioned Python code.}
First, we apply practical improvements for the KLL sketch proposed by Ivkin et al.~\cite{ivkin2019streaming}. 
These include laziness in compaction operations, i.e., allowing the buffer to exceed its capacity provided that the overall capacity of the sketch is satisfied,
and flipping a random coin for choosing even or odd indexed items only during every other compaction operation at each level and otherwise using the opposite outcome of the previous coin flip.

Furthermore, new ideas for \reqsketch are desired.
A specific feature of \reqsketch, compared to the KLL sketch, is that the buffer size $B$ depends on the input size $N$, and this is needed
for the prefix choice by the (derandomized) exponential distribution.
For the theoretical result, it is possible to maintain an upper bound $\hat{N}$ on $N$ and once it is violated, use $\hat{N}^2$ as an upper bound and recompute the buffer size at all levels.
As it turns out, it suffices for the exponential distribution
to count compaction operations performed at each level $h$ and set the level-$h$ buffer size based on this count $C_h$, i.e., to $O(\varepsilon^{-1}\cdot \sqrt{\log C_h})$.
This results in levels having different capacities, with lower levels being larger as they process more items.
Since the level-$0$ compactor has the largest size, compaction operations at level $0$ are most costly as they take time proportional to the size.
To improve the amortized update time, when we perform a compaction operation at level $0$, we also compact the buffer at any other level that exceeds its capacity.
In other words, we restrict the aforementioned laziness in compaction operations to level $0$ only and this postpones the next compaction operation at level $0$
for as long as possible.
Experiments reveal that such a ``partial laziness'' improves the amortized update time significantly (depending on the number of streaming updates and parameters);
see Section~\ref{sec:updateTime}.

Overall, the empirical results in this paper and in the DataSketches library suggest that on randomly shuffled inputs \reqsketch with the improvements outlined above
has over 10 times smaller standard deviation of the error than predicted by the (already quite tight) mathematical worst-case analysis in~\cite{cormode2020relative}.
Furthermore, results on various particular input orderings (performed in the DataSketches library and with the proof-of-concept Python implementation)
reveal that the error on random permutations is representative, that is, we did not encounter a data ordering on which \reqsketch has higher
standard deviation.

\section{Careful Attack on \tdigest}
\label{sec:attack}

\noindent
\paragraph{The \tdigest and space bounds.}
The application of merging on the centroids ensures that the number of centroids maintained by the \tdigest cannot be too large. 
In particular, the space parameter $\delta$ is used to enforce that there are at most $\delta$ centroids maintained, 
no matter what order the items in the input stream arrive in, for any of the scale functions considered above~\cite{dunning2019size,ross2020asymmetric}. 
While this bound on the size of the summary is reassuring, it appears to stand in opposition to the space lower bounds referenced in Section~\ref{sec:related} above~\cite{hungting10,cormode2019tight}. 
This conflict is resolved by observing that the lower bounds hold against algorithms which only apply comparisons to item identifiers, whereas the \tdigest combines centroids by taking weighted averages, and uses linear interpolation to answer queries.
Still, we should not be entirely reassured. 
We can consider a ``sampling" (instead of averaging) variant of the \tdigest approach which stays within the comparison-based model, by keeping some (randomly-chosen) item in each centroid as its representative, and using this to answer queries.  
For sure, this sampling variant makes less use of the information available,
and lacks the careful hierarchical structure of the KLL sketch or \reqsketch.
But now this sampling variant is certainly subject to the space lower bounds, e.g., from~\cite{cormode2019tight}, and so cannot guarantee accuracy while using only %
$O(\delta)$ space.  
Since it is not so different to the full averaging \tdigest, the accuracy offered might be comparable, suggesting that this too may be vulnerable. 
This is the essence of our subsequent study, and we make use of the construction of a ``hard'' instance from~\cite{cormode2019tight} to help form the adversarial inputs to \tdigest.

\paragraph{Overview of the attack.}
The idea of the attack is to produce a (very) weakly ordered collection of centroids by wielding inspiration from~\cite{cormode2019tight} against the inner workings of \tdigest. 
For motivation, if centroid $(c_1, w_1)$ summarizes $S_1:= \{ x_1, \ldots, x_{w_1}\}$ and $x_i < c_1 < x_{i+1} = \text{next} (\sigma, c_1)$ (where $\text{next} (\sigma, c_1)$ is the smallest stream element larger than $c_1$; see \cite{cormode2019tight}), then for a query point in the interval $(c_1, x_{i+1})$, the rank will be overestimated by (at least) $w_1 - i$. When $c_1$ is very close to the median of $S_1$, this produces an overestimate of approximately $\frac{w_1}{2}$. In the worst case, the rank error is closer to $w_1$: if we operate with integer weights, in the ``lopsided" scenario in which $x_2 = x_3 = \cdots = x_{w_1}$, the rank is overestimated by at least $w_1 - 1$ in the interval $(c_1, x_2)$. Using real-valued weights, the rank error can be made as close to $w_1$ as desired, using a weighted set of the form $\{ (x_1, \epsilon), (x_2, w_1 - \epsilon)\}$.

If a set $S_2:= \{ y_1, \ldots, y_{w_2}\}$ is then inserted within the interval $(c_1, x_{i+1})$ and forms a new centroid $(c_2, w_2)$, then for a query point in the interval $(c_2, \text{next} (\sigma, c_2))$, the rank will be overestimated by $\frac{w_1 + w_2}{2}$ in the typical case,  or $w_1 + w_2 - 2$ in the worst case with integer weights if $S_2$ is similarly lopsided; see Figure~\ref{fig:attack-illustration}. (The orientation of the attack may be flipped in the evident way, producing underestimates of the rank.)

\begin{figure}
    \centering
    \includegraphics[scale=1]{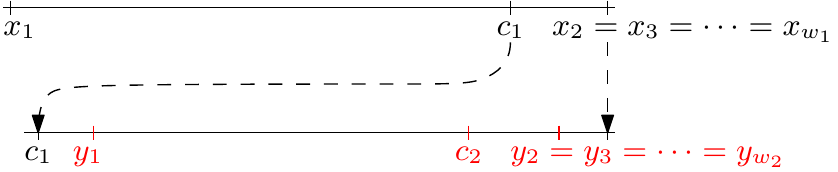}
    \caption{Illustration of the first two iterations of the attack that creates two lopsided centroids at $c_1$ and $c_2$,
    assuming for simplicity that none of items $y_1, \dots, y_{w_2}$ gets merged with the centroid at $c_1$.
    Centroids to the left or right of the centroid at $c_1$ are not shown.}
    \label{fig:attack-illustration}
\end{figure}

As this nested construction proceeds, some of the inserted points will be merged with a {previously created} centroid, and hence that portion of the inserted weight will not contribute to the rank error.
Overall, after adding an initial set of items, we choose a centroid $(c_1, w_1)$ that we will attack.
Assuming the merging pass always proceeds from left to right, the sequence of centroid weights progresses as:
\vspace{-1em}
\begin{equation} \begin{split} \label{insertions}
(- l_1 - ) , w_1, ( - r_1 -) \\
(- l_1 - ) , w_1 + v_1, w_2,  ( - r_1 -) \\
(- l_1 - ) , w_1 + v_1, w_2 + v_2 , w_3, ( - r_1 -)
\end{split} \end{equation}
and so on, 
where $(- l_1 - )$ denotes the ordered collection of centroids, having total weight $l_1$ and smaller than the centroid with weight $w_1$, and similarly $( - r_1 -)$ stands for centroids larger than the attacked centroid.
The idea is to add items of weight $v_1$ to the ``attacked'' centroid so that this centroid will be full and none of the next $w_2$ items will get merged into it, 
and similarly in the next iteration.
Thus, the first insertion has size $v_1 + w_2$, the second has size $v_2 + w_3$, etc. 
To see how this affects the rank error as the number of nested insertions of lopsided centroids increases, observe that if
\begin{equation} \label{infinite} \frac{\sum_{i=1}^N w_i + v_i}{l_1 + r_1} \to \infty \end{equation}
as $N \to \infty$ (i.e., if the weight not covered by these lopsided centroids is negligible),
then the asymptotic error can be made arbitrarily close to ${{\sum_i (w_i - 1)} \over {\sum_i w_i + v_i}}$ if integer weights are required, or ${{\sum_i w_i} \over {\sum_i w_i + v_i}}$ if this restriction is dropped.
If in addition
\begin{equation} \label{bounded_below}
{w_i \over {w_i + v_i}} \geq \gamma > 0
\end{equation}
for all $i$, then ${{\sum_i w_i} \over {\sum_i w_i + v_i}} \geq \gamma$ as well and hence $\gamma$ serves as a lower bound on the asymptotic error.

We will see that the parameter $\delta$ influences the rate of convergence to the asymptotic error (i.e., the growth of the quantity in \eqref{infinite}),
but the asymptotic error itself (i.e., $\gamma$ in \eqref{bounded_below}) cannot be reduced simply by taking $\delta$ large enough; in fact the asymptotic error is \textit{increasing} in $\delta$ for some important scale functions.
In the next sections we sketch how to achieve the inequalities \eqref{infinite} and \eqref{bounded_below} above for several scale functions of interest.

\subsection{Scale Functions with Bounded Derivative}
\begin{proposition}\label{prop:bounded_derivative}
Let $k$ be a scale function such that $0 < b \leq k'(q) \leq B$ for $q \in [0, 1]$. Then there exists a number $\gamma > 0$ and a $\delta_0 > 0$ such that for all $\delta \geq \delta_0$, the \tdigest associated with $(k, \delta)$ has asymptotic error at least $\gamma$ on the nested sequence of lopsided insertions described above.
\end{proposition}

\begin{proof}[Sketch of proof.]
Since by~\eqref{eqn:weight_limit}, the function $k$ increases by ${1 \over \delta}$ on the interval $I_w$ in quantile space occupied by a full centroid $(c, w)$, the Mean Value Theorem guarantees a point $q_w \in I_w$ such that $k'(q_w) |I_w| = {1 \over \delta}$, where $|I_w|$ denotes the length of the interval. 
Note that $|I_w| = w / (l + w + r)$, where $l$ and $r$ are the weights to the left and right of centroid $(c, w)$, respectively.
We apply this on centroid $(c_i, w_i + v_i)$ just after the $i$-th iteration of the attack, with $w = w_i + v_i$.
Since $\frac{w_i + v_i}{l_1 + r_1} \geq |I_w|$ and $k'(q_w) \leq B$, we obtain $\frac{w_i + v_i}{l_1 + r_1} \geq {1 \over {\delta B}}$. Taking $N$ large enough (depending on $\delta$), the desired limiting behavior \eqref{infinite} is shown.

For the second inequality, we apply similar arguments to the intervals of weights $w_i$, $w_i + v_i$, $w_{i+1}$ appearing in consecutive iterations of the attack. Clearing denominators, we obtain equations:
\begin{equation} \label{eqn:bounded derivative eqns} \begin{split}
\delta k'(q_{w_i}) w_i &= l + w_i + r \\
\delta k'(q_{w_i + v_i}) (w_i+ v_i) &= l + w_i + v_i + w_{i+1} + r \\
\delta k'(q_{w_{i+1}}) w_{i+1} &= l + w_i + v_i + w_{i+1} + r 
\end{split} \end{equation}
From which it follows that
\begin{equation}\label{eqn:w_i overl w_i+v_i}
    \frac{w_i}{w_i+v_i} = \frac{\delta k'(q_{w_i+v_i}) k'(q_{w_{i+1}}) - k'(q_{w_i+v_i}) - k'(q_{w_{i+1}})}{(\delta k'(q_{w_i}) - 1)k'(q_{w_{i+1}})}.
\end{equation}
The denominator is bounded above by $\delta B^2$. 
For any $\epsilon > 0$, we can find $\delta_0 > 0$ such that for $\delta \geq \delta_0$, the numerator is bounded below by $(1 - \epsilon) \delta b^2$.
Hence $\frac{w_i}{w_i+v_i}$ is bounded below by $\frac{(1 - \epsilon) b^2}{B^2}$ and \eqref{bounded_below} is shown as well.
\end{proof}

A consequence of the proof is that for $k_0(q) = q/2$ (the linear scale function producing roughly equal-weight centroids and expected to have constant accuracy), the asymptotic error according to~\eqref{eqn:w_i overl w_i+v_i} is $\frac{\delta - 2}{\delta - 1}$, i.e., for sufficiently large $\delta$, the approximations can be arbitrarily poor.

\subsection{Attacks on $k_2$ and $k_3$}%

Without loss of generality, we assume the attack occurs where $q > 0.5$
\textit{according to the \tdigest}, i.e., the attacked centroid has more weight to its left than to its right,
hence the growth conditions for $k_3$ under \eqref{insertions} give rise to the following system of equations (derived similarly to~\eqref{eqn:bounded derivative eqns}, %
using the definition of $k_3$):

\begin{equation} \label{k3_explicit} %
\exp\left( \frac{1}{ \delta} \right)
= \frac{w_i+r_1}{r_1}
= \frac{w_i+v_i+w_{i+1}+r_1}{w_{i+1} +r_1}
= \frac{w_{i+1}+r_1}{r_1}
\end{equation}

Solving yields $\frac{w_i}{w_i+v_i} = \frac{1}{\exp(\frac{1}{\delta})}$ and also $\frac{w_i}{r_1} = \exp(\frac{1}{\delta}) - 1$. We may assume $l_1 < C(\delta) r_1$ and hence $\frac{w_i}{l_1 + r_1} > \frac{w_i}{(C(\delta) + 1) r_1} = \frac{\exp(\frac{1}{\delta}) - 1}{(C(\delta) + 1)}$. From this the limiting behavior \eqref{infinite} follows; 
as $\delta \to \infty$, $\exp(\frac{1}{\delta}) \to 1^+$ and hence the asymptotic error gets larger (and approaches $1$) for larger values of $\delta$.
Hence, in the worst case, the quantile error of \tdigest with scale function $k_3$ can be arbitrarily close to 1.

While $k_2$ does not seem as amenable to direct calculation, we observe that for all $q \in (0.5, 1)$, we have $k_3'(q) < k_2'(q)  < 2 k_3'(q)$. Therefore the growth of $k_2$ on an interval can be bounded on both sides in terms of the growth of $k_3$, and the system of equations \eqref{k3_explicit} has a corresponding system of inequalities. Eventually we find $\frac{w_i}{r_1} >  \exp(\frac{1}{2 \delta}) -1$, giving \eqref{infinite}, and also get 
$$\frac{w_i}{w_i+v_i} > \frac{1}{\exp(\frac{1}{\delta}) ( \exp(\frac{1}{2 \delta}) + 1 )}.$$
This lower bound on the asymptotic error approaches $0.5$ as $\delta$ increases,
which implies that the quantile error of \tdigest with $k_2$ can be arbitrarily close to $0.5$.

\section{Empirical results}
\label{sec:empirical}

In this section, we study the error behavior of \tdigest and of \reqsketch on inputs constructed according to the ideas described in Section~\ref{sec:attack}
and also on inputs consisting of i.i.d.\ items generated by certain non-uniform distributions.
Furthermore, we compare the merging and clustering implementations of \tdigest, and also empirically evaluate update times of the two algorithms.

The experiments are performed using the Java implementation of \tdigest by Dunning
and the Java implementation of \reqsketch by the Apache DataSketches library.\footnote{
All code used is open source, and all scripts and code can be downloaded from our repository at \url{https://github.com/PavelVesely/t-digest/}.
For more details about reproducing our empirical results, see Section~\ref{sec:reproducibility}.}
By default, we run \tdigest with compression factor $\delta = 500$,
but we can obtain similar results for other values of $\delta$.
We then choose the parameters of \reqsketch so that its size is similar to that of \tdigest with $\delta = 500$ for the particular input size
(recall that the size of \reqsketch depends on the logarithm of the stream length).
The measure for the size we choose is the number of bytes of the serialized data structure.
For instance, if the input size is $N = 2^{20}$, then using $k=4$ as the accuracy parameter of \reqsketch leads to essentially the same size of the two sketches,
which is about $2.5$ kB.

\subsection{Implementation of the Attack}
\label{sec:attackImpl}

In implementing the ideas of Section \ref{sec:attack}, we note that the size of the interval between the attacked centroid and the next stream value shrinks exponentially as the attack proceeds. 
Hence the attack as described may run out of precision (at least for \texttt{float} or \texttt{double} variable types) after only a few iterations. 
To circumvent this difficulty, we target the attack in the neighborhood of zero, where more precision is available, so the attack as implemented\footnote{See again \url{https://github.com/PavelVesely/t-digest/}.} chooses the largest centroid less than zero, and uses the smallest positive stream value as its ``next'' element.
Additionally, the effectiveness of the attack can be sensitive to the exact compression schedule used by an implementation (particularly the clustering variant). 
Hence the results of the attack are somewhat dependent on the particular manner in which values are chosen from the interval for the ensuing iteration.
Nevertheless, equipped with knowledge of the parameters of the \tdigest, the ability to query for centroids near zero and centroid weights, and memory of the actual stream presented to the \tdigest, an adversary may generate a stream on which the \tdigest performs rather poorly. Figure~\ref{fig:careful_attack_errors} shows the (additive) quantile error %
of both the merging and clustering implementations of \tdigest, all using scale function $k_0$ and $\delta = 500$ (the error of \reqsketch is not shown as it is very close to 0\%, similarly as in the plots below). This shows that the vulnerability of \tdigest is not due to the specifics of implementation choices, but persists across a range of approaches.\footnote{A similar construction may be applied to $k_2$ or $k_3$, but as data accumulates on both sides of zero (for precision reasons), the error is not pushed into the tails. Higher precision computation (using, e.g., \texttt{BigDecimal}) would seem necessary for a practical implementation of the attack exhibiting poor performance in the tails of the distribution for the logarithmic scale functions.}

\begin{figure}[t]
    \centering
    \includegraphics[scale=0.6]{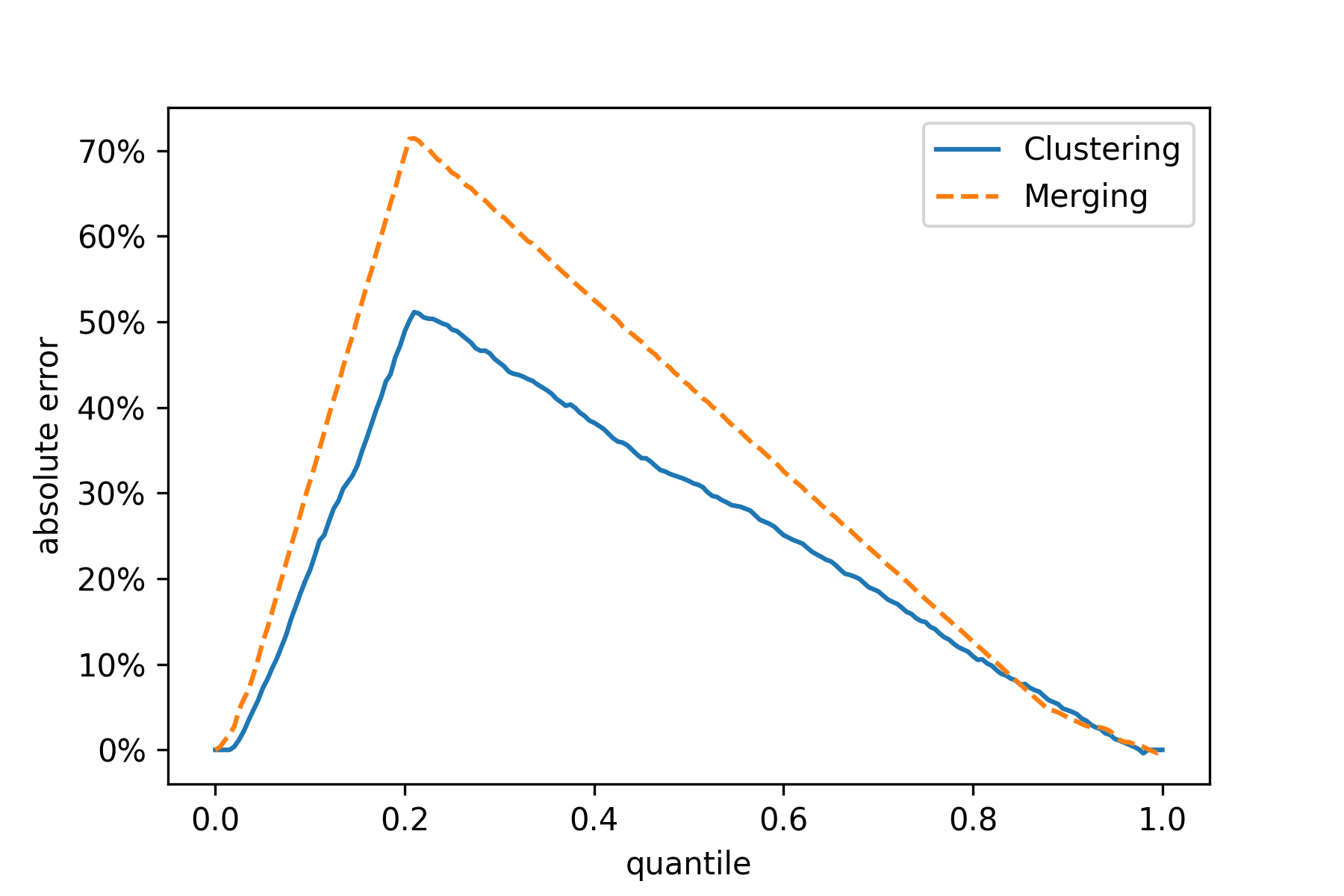}
    \caption{\tdigest on carefully constructed input}
    \label{fig:careful_attack_errors}
\end{figure}

\subsection{Randomly Generated Inputs}
\label{sec:iidSamples}

Here, we provide results for inputs consisting of i.i.d.\ items generated by some distributions.
Our purpose is to study the behavior of the algorithms when the order of item arrivals is not adversarially chosen, demonstrating that the class of ``difficult'' inputs is larger than just the carefully targeted attack stream. 
Items drawn i.i.d.\ form a well-understood scenario:
if we knew the description of the distribution, the quantiles of the distribution serve as accurate estimates of the quantiles of a sufficiently large input sampled from that distribution.
However, we only study algorithms not designed for any particular distribution, and so inputs consisting of i.i.d.\ items can still present a challenge.

It is already known that on (i.i.d.) samples from some distributions, such as the uniform or the normal distribution,
\tdigest performs very well~\cite{tdigest-origPaper,DataSketches-comparsion_t-digest_and_KLL,ross2020asymmetric}, and so we do not replicate these scenarios.
Instead, we study a class of distributions inspired by the attack of Section~\ref{sec:attack}
and show under which conditions \tdigest fails to provide any accurate rank estimates.

In the attack described in Section~\ref{sec:attack}, we carefully construct a sequence of insertions into \tdigest so that the error gets arbitrarily large.
Recall that with each iteration of the attack the interval where future items are generated shrinks exponentially,
while the number of items increases linearly.
Surprisingly, the large error is provoked not only by the \emph{order} in which items arrive but also by the large range of \emph{scales} of the input, as we demonstrate below.

Taking this property of the attack as an inspiration, a natural idea for a hard distribution is to generate items uniformly on a logarithmic scale, i.e.,
applying an exponential function to an input drawn uniformly at random.
This is called the \emph{log-uniform distribution}. 
To capture the increasing number of items in the iterations of the attack, 
we square the outcome of the uniform distribution used in the exponent and
finally, we let it have a negative or positive value with equal probability, giving a \emph{signed log-uniform${}^2$ distribution}. 
Thus, each item is distributed according to 
\begin{equation}\label{eqn:harddistribution}
    \harddistrib \sim (-1)^b\cdot 10^{(2\cdot R^2 - 1)\cdot E_{\max}}
\end{equation}
where $b\in \{0,1\}$ is a uniformly random bit, $R$ is a uniformly random number between $0$ and $1$, and
$E_{\max}$ is a maximum permissible exponent (for base~10) of the \texttt{double} data type, which is bounded by~308 in the IEEE 754-1985 standard.\footnote{
In our implementation, we use $E_{\max} = \log_{10} (M / N)$, where $M$ is the maximum value of \texttt{double},
so that \tdigest does not exceed $M$ when computing the average of any centroid.}
The input is then constructed by taking $N$ samples from $\harddistrib$.

Figure~\ref{fig:harddistribution_k_2_asym} shows the quantile error for \tdigest with $\delta = 500$ and asymmetric $k_2$ scale function together with the error of \reqsketch on this input with $N=2^{20}$.
We show the error for both the merging and clustering variants of \tdigest and take the median error of each rank, based on $2^{12}$ trials,
while for \reqsketch, we plot the 95\% confidence interval of the error, i.e., $\pm 2$ standard deviations
(recall that the error of \reqsketch for any rank is a sub-Gaussian zero-mean random variable).
The absolute error is plotted on the y-axis. 
The relative error requirement in this experiment is that the error should be close to zero for the high quantiles (close to 1.0), with the requirement relaxing as the quantile value decreases towards 0.0. 
This is observed for \reqsketch, whose error gradually increases in the range 1.0 to 0.5, before approximately plateauing below the median. 
However, the two \tdigest variants show larger errors on high quantiles, approaching -30\% absolute error at $q$=0.8 for the merging variant. 
Note that if we simply report the maximum input value for
$q=0.8$, this would achieve  +20\% absolute error.
The resulting size of the merging \tdigest is $2\,752$ bytes, while the clustering variant needs just $2\,048$ bytes and the size of \reqsketch is $2\,624$ bytes.

For comparison, Figure~\ref{fig:loguniform_k_2_asym} shows similar results on the signed log-uniform distribution, i.e., items generated
according to $(-1)^b\cdot 10^{(2\cdot R - 1)\cdot E_{\max}}$, with $b, R,$ and $E_{\max}$ as above.
The only notable difference is that \tdigest achieves a slightly better accuracy.
Our further experiments suggest that the error of \tdigest with other scale functions capturing the relative error, such as the (asymmetric) $k_3$ function,
is even worse than if asymmetric $k_2$ is used.
We obtain analogous results even for larger values of $\delta$ (with an appropriately increased parameter $k$ for \reqsketch),
although this requires larger stream lengths $N$.

A specific feature of both $\harddistrib$ and the log-uniform distribution is that there are numbers ranging (in absolute value) from $10^{-302}$ to $10^{302}$.
Such numbers, however, rarely appear in real-world datasets and one may naturally wonder what happens if we limit the range,
by bounding the parameter $E_{\max}$ of these distributions. %
Figure~\ref{fig:varying_E_max} shows the dependence of the \emph{average relative error} on $E_{\max}$\footnote{The relative error for
a particular rank $q$ is computed as the absolute quantile error at $q$ (i.e.\ the difference between $q$ and the estimated rank for item $y$ with $R(y) = q$)
divided by $1-q$; thus, the relative error is amplified for ranks close to $1$.
We run $2^8$ trials and for each rank we take the median relative error for \tdigest and both the median and $+2$ standard deviations (+2SD) of the error distribution for \reqsketch.
These median or +2SD errors are aggregated over all ranks by taking the average.};
a similar plot can also be obtained for the maximal relative error. %
Interestingly, the clustering variant of \tdigest performs far better than the merging variant w.r.t.\ this measure.
For a small $E_{\max}$, say, $E_{\max}\le 10$, \tdigest is clearly preferable to \reqsketch as a very large range of numbers is needed to enforce a large error for \tdigest.
On the other hand, as the error of \reqsketch is unbiased, its median relative error is indeed very close to 0.

We also note that the clustering variant of \tdigest appears to have better (though admittedly still inadequate) accuracy on the hard inputs than does the merging variant. The merging variant is generally preferred due to its faster updates (see Section~\ref{sec:updateTime}) and avoidance of dynamic allocations. Thus, those efficiencies may come with a price of higher error.

\begin{figure}[t]
    \centering
    \includegraphics[scale=0.6]{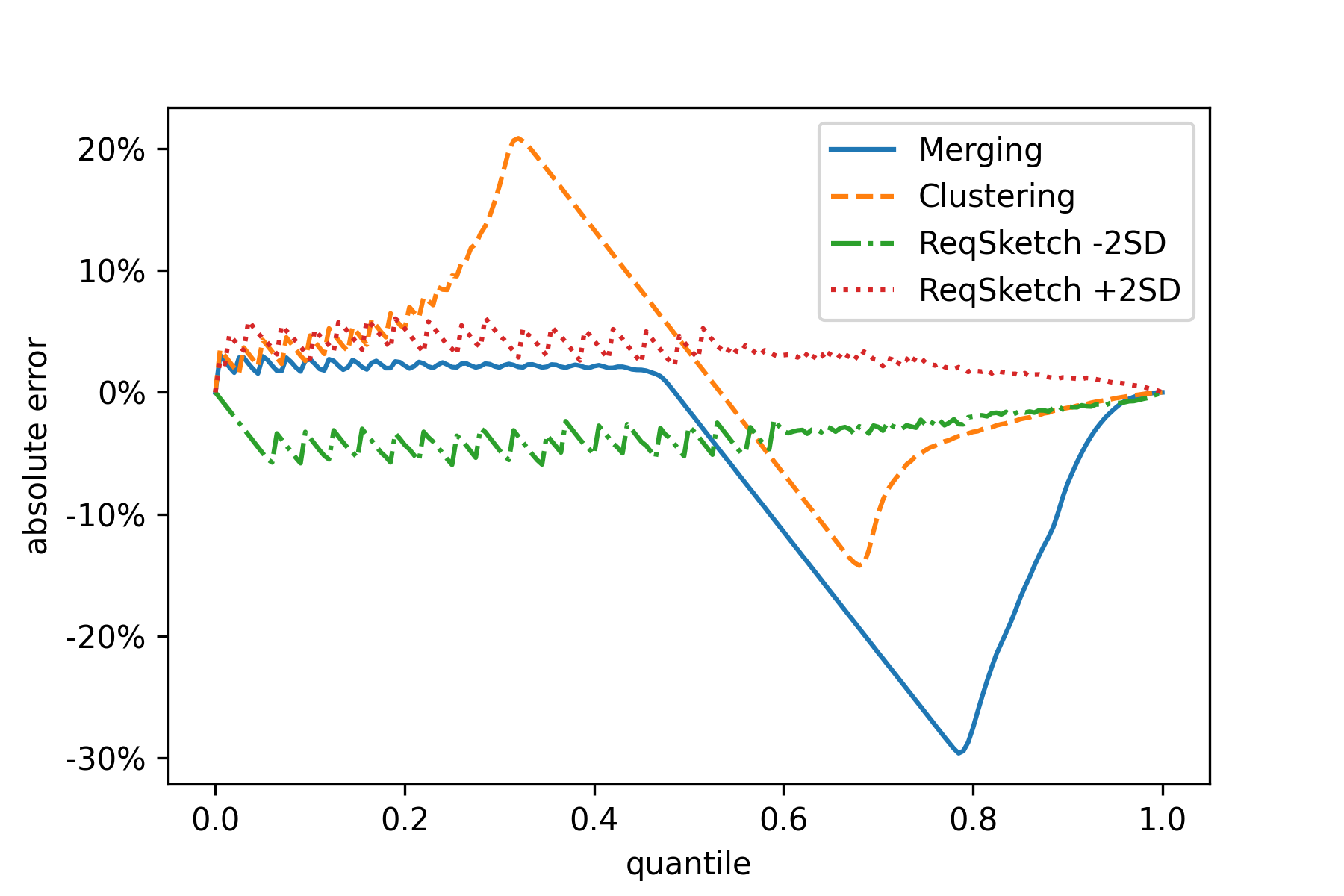}
    \caption{\tdigest on i.i.d.\ samples from $\harddistrib$ ($\pm 2$SD for \reqsketch means $\pm 2$ standard deviations)} %
    \label{fig:harddistribution_k_2_asym}
\end{figure}

\begin{figure}[t]
    \centering
    \includegraphics[scale=0.6]{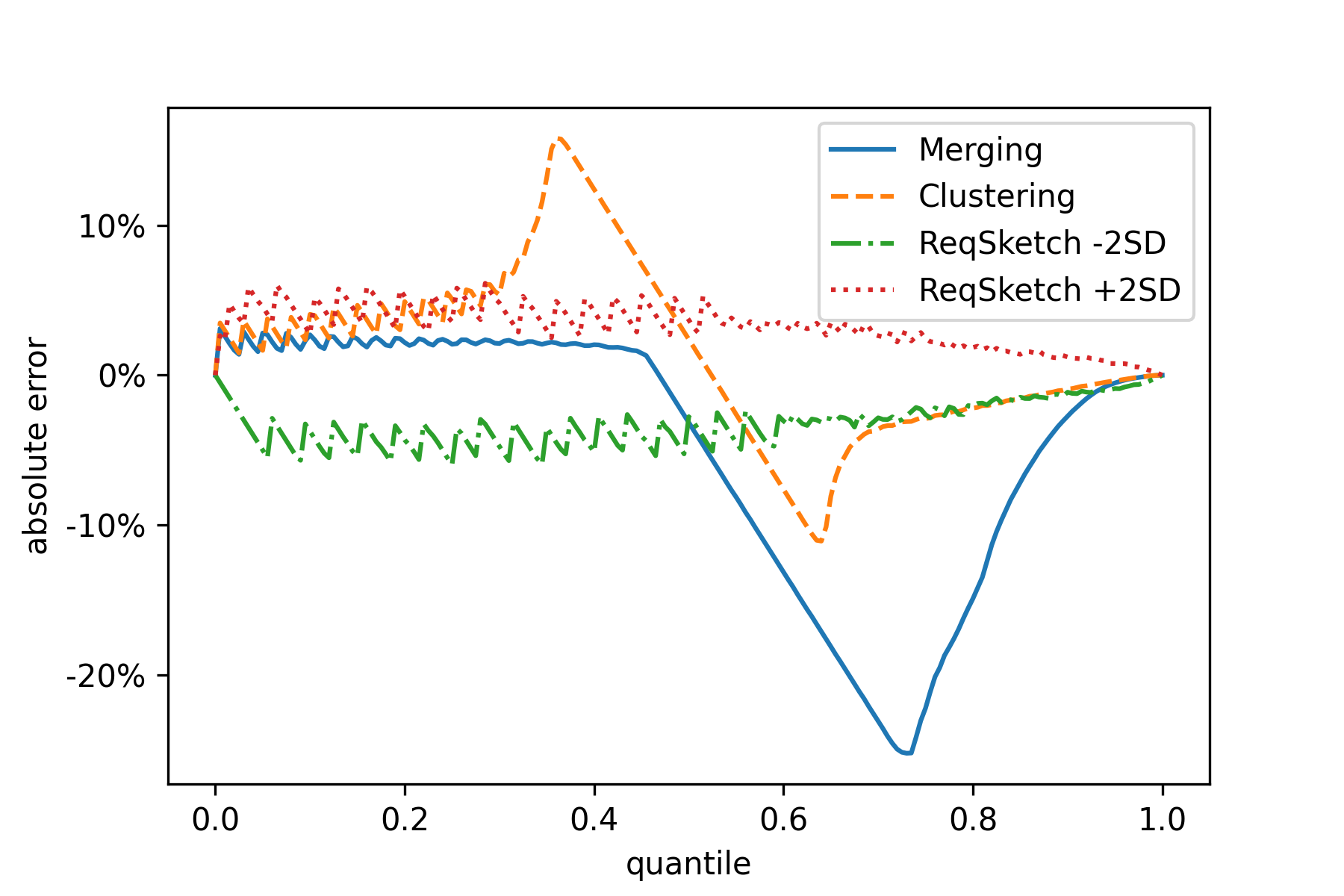}
    \caption{\tdigest on i.i.d.\ samples from the signed log-uniform distribution.}
    \label{fig:loguniform_k_2_asym}
\end{figure}

\begin{figure}[t]
    \centering
    \includegraphics[scale=0.6]{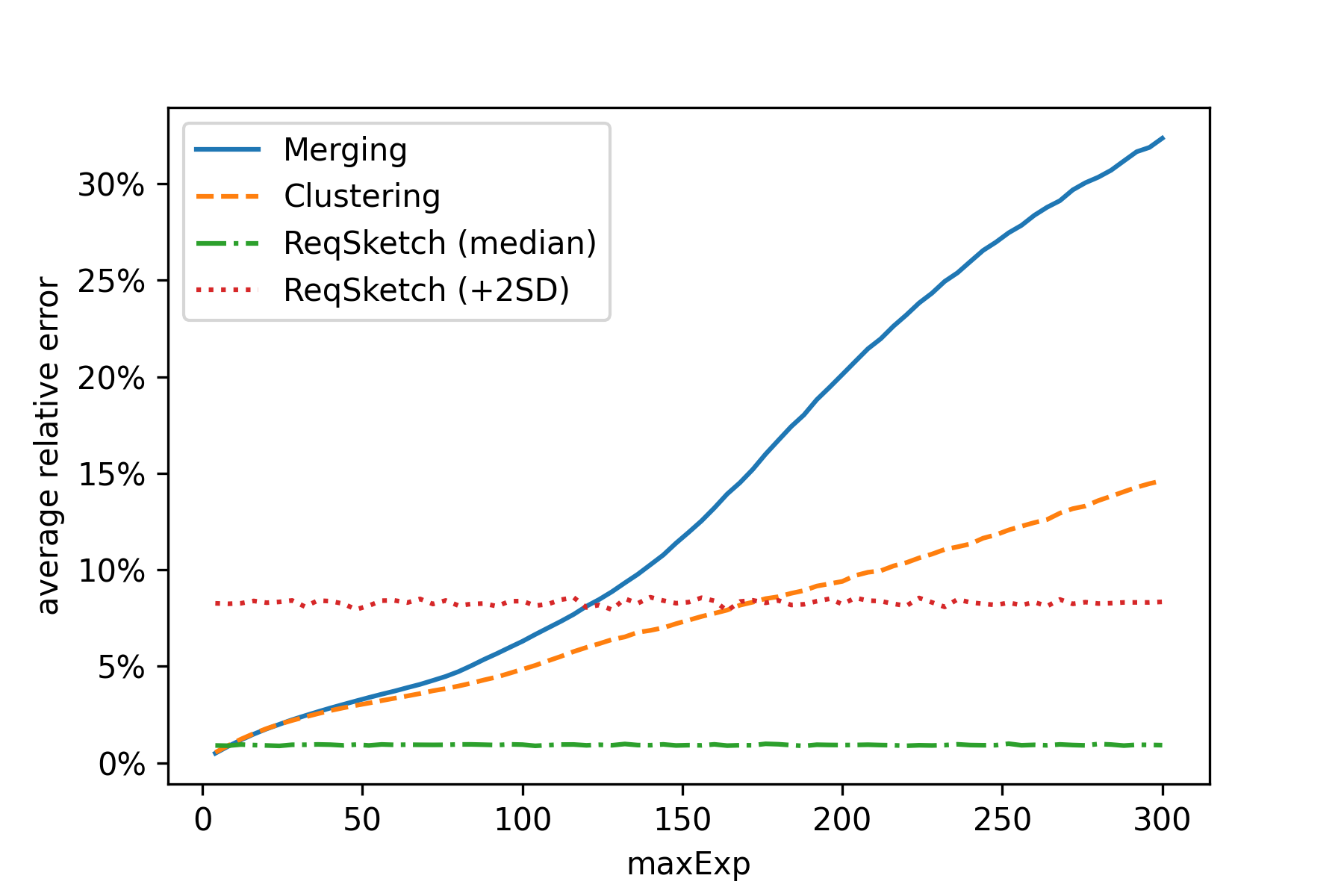}
    \caption{Average relative error on i.i.d.\ samples from $\harddistrib$, depending on $E_{\max}$ (denoted maxExp).} %
    \label{fig:varying_E_max}
\end{figure}

\paragraph{Explanation of the large error for \tdigest.}
The particularly striking feature of the merging \tdigest error in Figure~\ref{fig:harddistribution_k_2_asym} is the straight line which approximately goes from rank~$0.48$ to rank~$0.79$.
As it turns out, all ranks in this range have essentially the same estimated rank returned by \tdigest. 
This is because the last centroid with mean below $0$ has mean equal to $\approx -10^{-302}$,
while the next centroid in the order has mean $\approx +10^{-101}$ (for one particular trial). Hence, there is no centroid to represent values 
$[0, 10^{-101}]$, while according to the definition of the hard distribution in~\eqref{eqn:harddistribution},
approximately 40\% of items fall within that range. Furthermore, most of this 40\%  lies in a much smaller interval of, say, $[0, 10^{-110}]$,
meaning that linear interpolation does not help to make the estimated ranks more accurate.
A similar observation can be made about the error of the clustering variant, as well as in other scenarios. While the infinitesimal values are not well-represented by centroids, they distort the centroid means. In the clustering variant, for example, all the centroids are pulled towards zero by being averaged with infinitesimal items, leading to overestimates of quantiles for $q < 0.5$ and underestimates of quantiles for $q > 0.5$.\footnote{A similar phenomenon occurs for the merging variant, but the error shifts since merging passes always proceed from left to right. This is due to alternating merging passes not being properly supported for asymmetric scale functions.}

As outlined in Section~\ref{sec:tdigest}, centroids of \tdigest are only weakly ordered, which in particular means that the numbers covered by one centroid may be larger
than the mean of a subsequent centroid of \tdigest.
The mixed scales, when presented in random order, lead to centroids with a high degree of overlap, at least measured locally (on consecutive centroids): there are substantial regions in centroid weight space in which the two-sample Kolmogorov-Smirnov statistic computed on neighboring centroids is close to zero. The centroids produced by the careful attack, by contrast, are pairwise somewhat distinguishable, but have a global nested structure, causing the \tdigest to have large error.\footnote{See \url{https://github.com/PavelVesely/t-digest/blob/master/docs/python/adversarial_plots/notebooks/overlap_computation.ipynb} for the supporting computations.}
Note the same data presented \textit{in sorted order} does not pose nearly the same difficulty for \tdigest, as the infinitesimal items eventually form their own centroids and give the \tdigest sufficient detail on that scale.

\subsection{Update Time}
\label{sec:updateTime}

Finally, we provide empirical results that compare the running time of 
the open source Java implementations of \tdigest and \reqsketch.
We evaluate both merging and clustering variants of \tdigest associated with $\delta = 500$ and asymmetric $k_2$ scale function, 
and choose the accuracy parameter of \reqsketch as $k=4$.
Additionally, we include the KLL sketch as well, with accuracy parameter $k = 100$.
Table~\ref{tab:update_times} shows amortized update times in nanoseconds (rounded to integer) on an input consisting of
$N$ i.i.d.\ samples from a uniform distribution, for $N = 2^{30}$.
The results are obtained on a 3~GHz AMD EPYC 7302 processor.
We remark that the update times remain constant for varying $N$, unless $N$ is too small (in the order of thousands).
In summary, the results show \reqsketch to be more than two times faster than merging \tdigest and about $4.5$ times faster than clustering \tdigest.

\begin{table}[b]
    \centering
    \caption{Average update time in nanoseconds.}
    \label{tab:update_times}
   \begin{tabular}{c|c|c|c}
         Merging \tdigest & Clustering \tdigest & \reqsketch & KLL  \\
         \hline
         129 & 251 & 55 & 41
    \end{tabular}
 \end{table}

Furthermore, we compare \reqsketch with the partial laziness technique from Section~\ref{sec:reqSketchImprovements}
(which is the default option) and with ``full laziness'' that was proposed for the KLL sketch in~\cite{ivkin2019streaming}.
Figure~\ref{fig:update_times_reqsketch} shows the average update times of both variants for varying accuracy parameter $k$ on $N$ i.i.d.\ samples from a uniform distribution, for $N = 2^{30}$.
This implies that the partial laziness idea provides a significant speedup, especially for larger values of $k$.

\begin{figure}[t]
    \centering
    \includegraphics[scale=0.56]{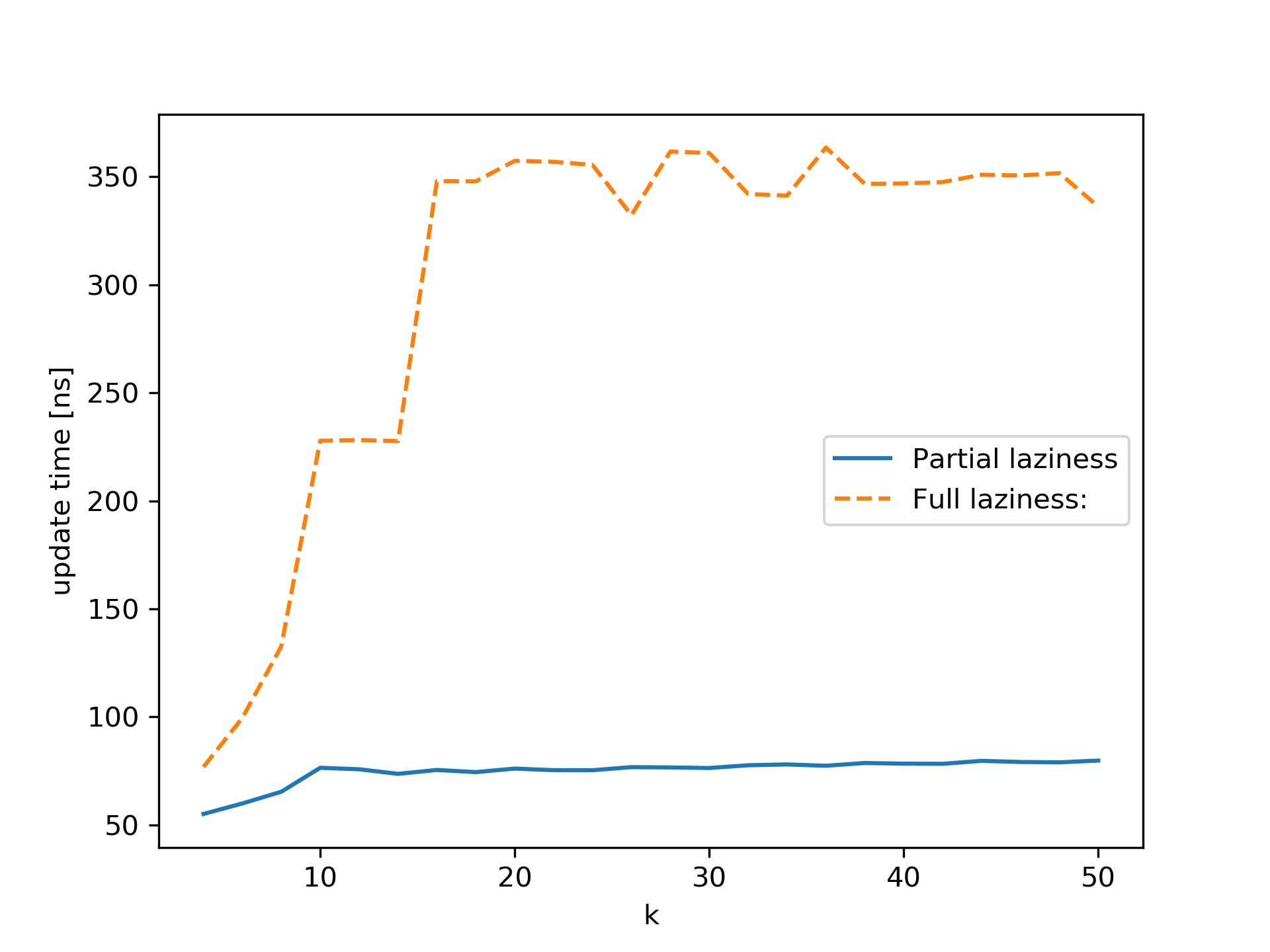}
    \caption{Average update time of \reqsketch in nanoseconds.}
    \label{fig:update_times_reqsketch}
\end{figure} 

\section{Concluding remarks}

Our standpoint as authors of this work cannot be viewed as neutral: 
two of us (G.\ Cormode and P.\ Veselý) are co-authors of the \reqsketch paper~\cite{cormode2020relative}, and two of us (at Splunk) have deployed \tdigest and analyzed its behavior. 
Our collaboration in this work was driven by a desire to better understand these algorithms, their strengths and weaknesses, and provide advice to other data scientists on how to make best use of them. 
As foreshadowed in the introduction, our view at this conclusion is perhaps more complicated than when we started, when we hoped for a simple answer.  
From our studies, the main takeaway is that \tdigest can fail to give the desired levels of accuracy on %
inputs with a highly non-uniform distribution over the domain.
However, these inputs are far from appearing natural, and should not significantly trouble any teams who have deployed this algorithm. 
Our second observation is that as implemented, the \reqsketch is pretty fast in practice, and quite reliable in accuracy, despite its somewhat offputtingly technical description.  
We did not observe any examples where its error shoots up, but on ``expected'' distributions (like the $\harddistrib$ with a small $E_{\max}$ in Figure~\ref{fig:varying_E_max}), it is appreciably less accurate than \tdigest. 
There is no clear win for the pragmatic or theoretically minded solutions, at least in this case. 
In the final analysis, our advice to practitioners is to consider both styles of algorithm for their applications, and to weigh up carefully the tradeoff between performance in the worst case to performance on average. 

\paragraph{Acknowledgements.}
We wish to thank Lee Rhodes and other people working on the DataSketches project for many useful discussions about implementing \reqsketch.
Work done while P.\ Vesel\'{y} was at the University of Warwick.
G. Cormode and P.\ Vesel\'{y} were supported by 
European Research Council grant ERC-2014-CoG 647557.
P.\ Vesel\'{y} was also partially supported by GA \v{C}R project 19-27871X and by Charles University project UNCE/SCI/004.

\bibliographystyle{plain}

\appendix
\section{Reproducibility}
\label{sec:reproducibility}

All code used in obtaining the experimental results is open source and can be downloaded from 
\begin{equation*}
    \text{\url{https://github.com/PavelVesely/t-digest/}}\,,    
\end{equation*}
where we also provide documentation and resources needed to reproduce our experiments.
Our repository is a clone of the original \tdigest repository available at
\url{https://github.com/tdunning/t-digest} (the original repository was last merged into ours on 2021-01-28),
and it additionally incorporates asymmetric scale functions from \url{https://github.com/signalfx/t-digest/tree/asymmetric}. The asymmetric scale functions provide a natural \tdigest analogue of a \reqsketch with guarantees on one end of the distribution.

The DataSketches library is available at \url{https://datasketches.apache.org/}
and we took the particular Java implementation of \reqsketch from the GitHub repository at \url{https://github.com/apache/datasketches-java}.
For technical reasons, the code we use in our experiments requires the \reqsketch algorithm to work with the \texttt{double} data type,
however, the DataSketches implementation works with \texttt{float} numbers only. We provide an adjusted implementation using \texttt{double}
inside our above-mentioned repository for reproducing the experiments.
We also incorporate the KLL sketch from the DataSketches library into our repository, with a similar adjustment to the \texttt{double} type.

\subsection{Main Experimental Setups}

We implemented three experimental setups:
\begin{itemize}
    \item A careful construction of a hard input for \tdigest, according to Sections~\ref{sec:attack} and~\ref{sec:attackImpl}. 
    \item A generator of i.i.d.~samples from a specified distribution, for reproducing results in Section~\ref{sec:iidSamples}.
    A variant of this experiment allows to have variable parameter $E_{\max}$ of $\harddistrib$ and the log-uniform distribution, for reproducing results in Figure~\ref{fig:varying_E_max}.
    \item A comparison of the average update times of \tdigest (both the merging and clustering variants), \reqsketch,
    and the KLL sketch, for reproducing results in Section~\ref{sec:updateTime}.
\end{itemize}

The parameters of these experiments are adjustable by a configuration file, which allows, for example,
to set the compression parameter $\delta$ and scale function for \tdigest and the accuracy parameter $k$  for \reqsketch.
Each of the experiments outputs a CSV file with results into a specified directory. 
See the README file in the repository for more details on how to run the experiments and how to produce the plots.

The first two experiments output statistics on absolute error of \tdigest and of \reqsketch for each of
200~evenly spaced normalized ranks (the number of these ranks can be adjusted). Furthermore, we perform $T$ trials, where $T$ is adjustable and set to $2^{12}$ by default,
and output the median error for each variant of \tdigest and the 95\% confidence interval for \reqsketch (recall that the error of \reqsketch is unbiased; see Section~\ref{sec:reqsketch}). 
More precisely, the errors for each rank are accumulated using the KLL sketch with accuracy parameter $k=200$ and then 
from this sketch we recover an approximate median or appropriate quantiles for two standard deviations of the normal distribution.
The error introduced by using the KLL sketch instead of exact quantiles is negligible as we do not need to estimate extreme quantiles of the distribution.
The experiment with variable $E_{\max}$ (for Figure~\ref{fig:varying_E_max}) outputs average and maximal relative errors\footnotemark[6] of \reqsketch and both variants of \tdigest
for each tested value $E_{\max}$.

\subsection{Auxiliary Experiments}

The possibility to adjust the parameters allows for verification of other claims in this paper. For instance, one can obtain plots similar to those in~Figures~\ref{fig:harddistribution_k_2_asym}-\ref{fig:varying_E_max} for (asymmetric) scale function $k_3$ or for other values of $\delta$.
We remark that in general, larger values of $\delta$ require larger values of the input size $N$ to induce poor accuracy levels for \tdigest, compared to similarly-sized \reqsketch.

Additionally, the configuration files may be altered to produce more verbose output, namely to also write the datapoints underlying the centroids in the resulting \tdigest. Some plots describing the local overlap of centroids are available in the repository. These help to illuminate the nature of the weak ordering of centroids discussed in Section \ref{sec:iidSamples}.

Finally, further experiments with \reqsketch can be performed with our proof-of-concept Python implementation and in the DataSketches library.
The Python implementation of \reqsketch by the fourth author is available at \url{https://github.com/edoliberty/streaming-quantiles/blob/master/relativeErrorSketch.py}
and the generator of some particular data orderings is at \url{https://github.com/edoliberty/streaming-quantiles/blob/master/streamMaker.py}.
Moreover, the DataSketches library provides a repository for doing extensive accuracy and speed experiments with \reqsketch (as well as other sketches in this library),
which is available at \url{https://github.com/apache/datasketches-characterization/}.

\end{document}